\documentclass[conference]{IEEEtran}

\IEEEoverridecommandlockouts
\usepackage{graphicx}
\graphicspath{{./images/}}
\usepackage[caption=false,font=footnotesize]{subfig}
\usepackage[dvipsnames]{xcolor}
\usepackage[unicode, psdextra, bookmarks, colorlinks, breaklinks]{hyperref}
\hypersetup{
    linkcolor=OliveGreen,
    citecolor=OliveGreen,
    filecolor=Black,
    urlcolor=MidnightBlue
}

\newif\ifcolor
\colorfalse

\newif\iflong
\longtrue

\usepackage[style=ieee, giveninits=true, maxnames=99, doi=false]{biblatex}
\usepackage{amsmath, amssymb}
\usepackage{physics}
\usepackage{xargs}
\usepackage{xifthen}
\usepackage{dsfont}

\usepackage[nameinlink,capitalize]{cleveref}
\usepackage[amsmath,standard,thmmarks,hyperref,thref]{ntheorem}

\theoremclass{theorem}
\theoremseparator{.}
\renewtheorem{theorem}{Theorem}
\renewtheorem{lemma}{Lemma}
\renewtheorem{corollary}{Corollary}
\renewtheorem{proposition}{Proposition}

\theoremclass{definition}
\theoremseparator{.}
\renewtheorem{definition}{Definition}

\theoremclass{example}
\theoremseparator{.}
\renewtheorem{example}{Example}

\theoremclass{proof}
\theoremseparator{:}
\renewtheorem{proof}{Proof}

\newcommand{\N}{\mathbb{N}}
\renewcommand{\G}{\mathcal{G}}

\DeclareMathOperator*{\E}{\mathbb{E}}
\newcommand{\Exp}[2][]{\ifthenelse{\isempty{#1}}{\E[#2]}{\E_{#1}[#2]}}
\newcommand{\Expl}[2][]{\ifthenelse{\isempty{#1}}{\E[#2]}{\E\nolimits_{#1}[#2]}}

\newcommand{\dtv}{\operatorname{d_{\mathrm{TV}}}}
\newcommand{\argmin}{\operatorname*{\arg\min}}

\title{A Non-Asymptotic Analysis\\of Mismatched Guesswork\\%
\thanks{This work was supported by DARPA Grant HR00112120008.}
}

\author{
\IEEEauthorblockN{
Alexander Mariona\IEEEauthorrefmark{1},
Homa Esfahanizadeh\IEEEauthorrefmark{1},
Rafael G.~L.~D'Oliveira\IEEEauthorrefmark{2}, and
Muriel M\'{e}dard\IEEEauthorrefmark{1}}

\IEEEauthorblockA{
\IEEEauthorrefmark{1}
Research Laboratory of Electronics,
Massachusetts Institute of Technology, Cambridge, MA 02139}

\IEEEauthorblockA{
\IEEEauthorrefmark{2}
School of Mathematical and Statistical Sciences,
Clemson University, Clemson, SC 29634\\
Emails: amariona@mit.edu, homaesf@mit.edu, rdolive@clemson.edu, medard@mit.edu}
}

\begin{document}

\maketitle

\begin{abstract}
The problem of mismatched guesswork considers the additional cost incurred by
using a guessing function which is optimal for a distribution $q$ when the
random variable to be guessed is actually distributed according to a different
distribution $p$. This problem has been well-studied from an asymptotic
perspective, but there has been little work on quantifying the difference in
guesswork between optimal and suboptimal strategies for a finite number of
symbols. In this non-asymptotic regime, we consider a definition for mismatched
guesswork which we show is equivalent to a variant of the Kendall tau
permutation distance applied to optimal guessing functions for the mismatched
distributions. We use this formulation to bound the cost of guesswork under
mismatch given a bound on the total variation distance between the two
distributions.
\end{abstract}

\section{Introduction}\label{sec:intro}

Let $X$ be a random variable taking values in $\{1,\dots,n\}$ according to the
distribution $p$. Suppose that we would like to try and guess the realization
of $X$ by sequentially guessing possible values according to a guessing
function $G$, which is some permutation on $\{1,\dots,n\}$, stopping when we
correctly guess the value of $X$. The guesswork associated with $G$ is
$\Exp{G(X)}$. This is the general setting of the theory of guesswork
\cite{Mas94,Ari96}.

The goal of this work is to characterize the difference in guesswork between
the guessing function which is optimal for the generating distribution $p$ and
that which is optimal for a different distribution $q$. This problem is
generally referred to as \emph{guessing under mismatch}. We define the expected
cost of guesswork under mismatch between $p$ and $q$ to be
\begin{equation*}
    \delta(p,q) = \min_{G_q} \Exp[X\sim p]{G_q(X)-G_p(X)},
\end{equation*}
where $G_p$ and $G_q$ are optimal guessing functions for $p$ and $q$
respectively. Our approach is non-asymptotic and our methods are combinatorial:
we study the cost of mismatch for a fixed, finite value of $n$ and we analyze
guessing functions as permutations. This viewpoint emphasizes the structure of
the space of probability distributions in the context of guesswork. The main
contributions of this work are (1) an equivalence between the cost of mismatch
and a variant of the Kendall tau permutation distance \cite{Ken38,KV10} and (2)
a bound on the cost of mismatch given a bound on the total variation distance
between the two distributions.

\subsection{Related Work}
In line with the majority of the literature on guesswork, the problem of
guessing under mismatch has generally been approached from an asymptotic
perspective. The problems of mismatched guesswork, universal guessing, guessing
subject to distortion, and other settings have been well-studied in that
asymptotic framework. We describe here only a selection of this work,
describing its goals and highlighting how our problem setting and results
differ.

Arikan and Merhav considered the problem of guessing to within some distortion
measure and propose asymptotically optimal guessing schemes for such a measure
\cite{AM98}. Sundaresan considered the problem of guessing over a distribution
which is a member of a known family. To this end, they defined a notion of
redundancy which quantifies the increase in the moments of guesswork when using
a suboptimal guessing function \cite{Sun06}. Both of these works propose
universal strategies for certain families of sources, where ``universal'' means
that some quantity asymptotically approaches its natural limit, e.g., the
guesswork growth exponent for a particular moment approaches its minimum over
the family, or the normalized redundancy approaches zero. Salamatian et
al.~consider the asymptotic behavior of mismatched guesswork by means of large
deviation principles \cite{SLB+19}, a theory which has found considerable
applications in the context of guesswork \cite{HS10,CD12,BCC+18}.  They treat
independent, identically distributed sources using the framework of tilted
distributions developed in \cite{BCC+15,BCC+18} and derive an expression for
limit of the average growth rate of the moments of mismatched guesswork.

Our work distinguishes itself from the above in two key ways. First and
foremost, we are not interested in asymptotic equivalences or limiting
behavior. We view guessing functions as permutations on a finite alphabet
through a framework which highlights applications of combinatorial techniques
to guesswork analysis. Second, the bounds and limits which are given in, e.g.,
\cite{AM98}, \cite{Sun06} and \cite{SLB+19} involve the difference in the
exponents of guesswork, whereas we consider the difference in guesswork
directly. Their bounds are not generally applicable to our problem setting and
cannot be directly compared.

The remainder of this paper is organized as follows. \cref{sec:prelims}
establishes our definitions. In \cref{sec:permutations} we express the expected
cost of guesswork under mismatch in terms of a variation of the Kendall tau
permutation distance. We then use this formulation of guesswork in
\cref{sec:statistical} to bound the cost of mismatch in terms of the total
variation distance between the distributions in question.
\cref{sec:conclusion} offers concluding remarks and notes potential directions
for future work.

\section{Preliminaries}\label{sec:prelims}

In this section, we review and establish definitions for guesswork,
permutations, and statistical metrics. All probability distributions are over a
finite space of cardinality $n\in\N$, which we denote by $[n]=\{1,2,\dots,n\}$.
We refer to bijections from $[n]$ to $[n]$ as either permutations or guessing
functions.  Explicitly, if $G:[n]\to[n]$ is a guessing function, then $G(i)$ is
equal to the number of guesses used to guess element $i$. We refer to
$\Expl[p]{G}=\Expl[X\sim p]{G(X)}$ as the guesswork of $G$ under $p$. The
operator $\circ$ denotes function composition (we limit its use to
permutations).

A natural set of guessing functions to consider is that consisting of those
which minimize the guesswork.
\begin{definition}
The set of \emph{optimal guessing functions} for the distribution $p$ is
defined to be $\G_p = \argmin_{G} \Expl[p]{G}$.
\end{definition}
We denote by $G_p$ an optimal guessing functions for $p$. If $\abs{\G_p}=1$,
then we may unambiguously refer to \emph{the} optimal guessing function for
$p$. A simple observation first made by Massey \cite{Mas94} is that optimal
guessing functions proceed in order of decreasing probability, i.e.,
$G_p(i)<G_p(j)$ if and only if $p_i\geq p_j$. It follows that there are
multiple optimal guessing functions for $p$ if and only if $p$ assigns the same
probability to more than one point.

One contribution of this paper is a connection between guesswork and
permutations via a metric. A classical metric on permutations is the Kendall
tau rank distance \cite{Ken38}, which is equal to the minimum number of
adjacent transpositions required to transform one permutation to another. An
adjacent transposition $\tau$ is a permutation which exchanges two consecutive
elements, i.e., $\tau(j)=j+1$ and $\tau(j+1)=j$, where $1\leq j < n$, and
$\tau(i)=i$ for all $i\notin\{j,j+1\}$.
\begin{definition}
The \emph{Kendall tau distance} between two permutations $\sigma_1$ and
$\sigma_2$ is defined to be
\begin{equation*}
    K(\sigma_1,\sigma_2) = \sum_{\substack{(i,j):\\\sigma_1(i)<\sigma_1(j)}}
        \mathds{1}\{\sigma_2(i)>\sigma_2(j)\}.
\end{equation*}
\end{definition}
We say that two permutations $\sigma_1$ and $\sigma_2$ differ by an adjacent
transposition if there exists an adjacent transposition $\tau$ such that
$\sigma_1=\tau\circ\sigma_2$. The Kendall tau distance has been generalized in
many different ways to account for different notions of distance based on the
particular applications \cite{KV10}. In \cref{sec:permutations} we introduce a
version which we show describes the difference in guesswork between two
guessing functions.

Our second main result connects the statistical distance between two
distributions and the difference in optimal guesswork for those distributions.
The statistical metric we consider is the total variation distance, which
is proportional to the $L^1$ norm for distributions over finite spaces.
\begin{definition}
The \emph{total variation distance} between two distributions $p$ and $q$ over
$[n]$ is defined to be
\begin{equation*}
    \dtv(p,q) = \frac{1}{2}\sum_{i=1}^n \abs{p_i-q_i}.
\end{equation*}
\end{definition}

\section{Guesswork and a Permutation Divergence}\label{sec:permutations}

We begin by defining our notion of difference in guesswork.
\begin{definition}\label{def:list-cost}
Let $G_1$ and $G_2$ be guessing functions. The \emph{expected cost} of $G_2$
over $G_1$ with respect to the distribution $p$ is defined to be
$\Delta_p(G_1,G_2) = \Expl[p]{G_2-G_1}$.
\end{definition}
\begin{definition}\label{def:dist-cost}
The \emph{expected cost of guesswork under mismatch} between the distributions
$p$ and $q$ is defined to be
\begin{equation*}
    \delta(p,q) = \min_{G_q\in\G_q} \Delta_p(G_p,G_q),
\end{equation*}
where $G_p$ is any optimal guessing function for $p$.
\end{definition}
\iflong
The expected cost of guesswork under mismatch is defined by a minimum over
$\G_q$ because it is not necessarily true that $\Expl[p]{G_q}=\Expl[p]{G_q'}$
for all $G_q,G_q'\in\G_q$ (see \cref{sec:min}).
\fi
However, the minimization in \thref{def:dist-cost} need not be taken over
$\G_p$ because, by definition, $\Expl[p]{G_p}=\Expl[p]{G_p'}$ for all
$G_p,G_p'\in\G_p$. By taking this minimum, \thref{def:dist-cost} gives the
``best-case'' cost, but we emphasize that we do not formulate the expected cost
of guesswork under mismatch in terms of particular guessing functions, but
rather distributions.  With respect to expected guesswork, we are indifferent
to the choice among multiple optimal guessing functions.

We show that $\Delta_p(G_1,G_2)$ is equivalent to a version of the Kendall tau
distance which weighs pairs of transposed elements by the difference in
their probabilities under $p$.
\begin{definition}
The \emph{probability-weighted Kendall tau signed divergence} between two
permutations $\sigma_1$ and $\sigma_2$ with respect to a distribution $p$ is
defined to be
\begin{equation*}
    K_p(\sigma_1,\sigma_2) = \sum_{\substack{(i,j):\\\sigma_1(i)<\sigma_1(j)}}
        (p_i-p_j)\cdot\mathds{1}\{\sigma_2(i)>\sigma_2(j)\}.
\end{equation*}
\end{definition}
Under this divergence, permutations which differ on elements of similar
probability are considered close. Informally, the sign of the divergence is
positive if, in total, the transposition from $\sigma_1$ to $\sigma_2$ moves
elements with higher probability to later positions and elements with lower
probability to earlier positions. This notion is formalized in the equivalence
shown in \thref{th:pweight}. This quantity is not a metric, as it is
signed and asymmetric, but these properties are desirable in this setting.

A useful result on this divergence is that it satisfies a kind of ``triangle
equality.'' This holds for any choice of three permutations, but we limit the
setting of the following lemma for simplicity. Only this weaker version is
necessary to prove \thref{th:pweight}, which itself implies the general result.

\begin{lemma}\label{lm:additive}
Let $\sigma_1$, $\sigma_2$, and $\sigma_3$ be permutations such that $\sigma_2$
and $\sigma_3$ differ by an adjacent transposition. Then, for all distributions
$p$,
\begin{equation*}
    K_p(\sigma_1,\sigma_3)=K_p(\sigma_1,\sigma_2)+K_p(\sigma_2,\sigma_3).
\end{equation*}
\end{lemma}
\begin{proof}
Let $\tau$ be the transposition such that $\sigma_3=\tau\circ \sigma_2$ and let
$x$ and $y$ be the elements whose positions $\tau$ inverts, such that
$\sigma_2(x)=\sigma_2(y)-1$ and $\sigma_3(y)=\sigma_3(x)-1$. Then, since
$\sigma_2(z)=\sigma_3(z)$ for all $z\notin\{x,y\}$, we can write
\begin{align*}
    K_p(\sigma_1,&\sigma_3) =
    \sum_{\substack{(i,j):\\\sigma_1(i)<\sigma_1(j)}}
    (p_i-p_j)\cdot\mathds{1}\{\sigma_3(i)>\sigma_3(j)\} \\
\begin{split}
    ={}& \sum_{\substack{(i,j)\notin\{(x,y),(y,x)\}:\\\sigma_1(i)<\sigma_1(j)}}
        (p_i-p_j)\cdot\mathds{1}\{\sigma_2(i)>\sigma_2(j)\} \\
    &{}+ (p_x-p_y)\cdot\mathds{1}\{\sigma_1(x)<\sigma_1(y)\land
        \sigma_3(x)>\sigma_3(y)\} \\
    &{}+ (p_y-p_x)\cdot\mathds{1}\{\sigma_1(y)<\sigma_1(x)\land
        \sigma_3(y)>\sigma_3(x)\},
\end{split}
\end{align*}
where $\land$ denotes the intersection of events. By assumption,
$\sigma_3(x)>\sigma_3(y)$, so the term with a factor of $(p_y-p_x)$
is always zero, and we can write the term with a factor of $(p_x-p_y)$ as
\begin{align*}
    &(p_x-p_y)\cdot\mathds{1}\{\sigma_1(x)<\sigma_1(y)\land
        \sigma_3(x)>\sigma_3(y)\} = \\
    &(p_x-p_y)\cdot\mathds{1}\{\sigma_1(x)<\sigma_1(y)\} = \\
    &(p_x-p_y)
    + (p_x-p_y)\cdot\mathds{1}\{\sigma_1(x)<\sigma_1(y)\land
        \sigma_2(x)>\sigma_2(y)\} \\
    &\phantom{(p_x-p_y)} +
    (p_y-p_x)\cdot\mathds{1}\{\sigma_1(y)<\sigma_1(x)\land
    \sigma_2(y)>\sigma_2(x)\}.
\end{align*}
The second term of the final expression  is always zero since
$\sigma_2(x)<\sigma_2(y)$ by assumption. The third term is zero if
$\sigma_1(x)<\sigma_1(y)$ but cancels the constant $(p_x-p_y)$ otherwise.
Finally, it is easy to verify that $K_p(\sigma_2,\sigma_3)=p_x-p_y$.
Substituting,
\begin{align*}
\begin{split}
    K_p(\sigma_1,\sigma_3) ={}&
    \sum_{\substack{(i,j):\\\sigma_1(i)<\sigma_1(j)}}
        (p_i-p_j)\cdot\mathds{1}\{\sigma_2(i)>\sigma_2(j)\} \\
    &{\quad}+ K_p(\sigma_2,\sigma_3)
\end{split} \\
    ={}& K_p(\sigma_1,\sigma_2) + K_p(\sigma_2,\sigma_3).
\end{align*}
\end{proof}

The main result of this section is that the expected cost of guesswork under
mismatch can be formulated in terms of the probability-weighted Kendall tau
signed divergence.

\begin{theorem}\label{th:pweight}
If $G_1$ and $G_2$ are guessing functions and $p$ is a distribution, then the
expected cost of $G_2$ over $G_1$ with respect to $p$ is equal to the
probability-weighted Kendall tau signed divergence between $G_1$ and $G_2$ with
respect to $p$, i.e., $\Delta_p(G_1,G_2)=K_p(G_1,G_2)$.
\end{theorem}

\begin{proof}
Let $\sigma$ be the permutation which takes $G_1$ to $G_2$, i.e.,
$G_2=\sigma\circ G_1$. Since any permutation can be written as the composition
of a finite sequence of adjacent transpositions, for some $M\in \N$ we can
write $\sigma=\tau_M\circ \tau_{M-1}\circ \dots\circ \tau_1$, where $\tau_i$ is
an adjacent transposition for all $i\in[M]$. This gives a finite sequence of
permutations $G_1=\sigma_0,\sigma_1,\dots,\sigma_M=G_2$ such that
$\sigma_i=\tau_i\circ \sigma_{i-1}$. Let $y_i$ and $x_i$ be the \emph{original}
elements which $\tau_i$ moves one position up and down respectively. In
particular, $x_i$ and $y_i$ are defined such that
\begin{align*}
    \sigma_i(x_i) &= \sigma_{i-1}(y_i) = \sigma_{i-1}(x_i)+1, \\
    \sigma_i(y_i) &= \sigma_{i-1}(x_i) = \sigma_{i-1}(y_i)-1, \\
    \sigma_i(z) &= \sigma_{i-1}(z), \quad \forall z\notin\{x_i,y_i\}.
\end{align*}

We proceed by induction on $M$. If $M=1$, then we can write $G_2=\tau_1\circ
G_1$. Since $\tau_1$ is an adjacent transposition which inverts the positions
of $x_1$ and $y_1$ as defined above,
\begin{align*}
    \Delta_p(G_1,G_2) ={}& \sum_{k=1}^n p_k \cdot [G_2(k) - G_1(k)] \\
\begin{split}
    ={}& p_{x_1} \cdot [G_2(x_1) - G_1(x_1)] \\
    &{}+ p_{y_1} \cdot [G_2(y_1) - G_1(y_1)]
\end{split} \\
    ={}& p_{x_1} - p_{y_1} \\
    ={}& \sum_{\substack{(i,j):\\G_1(i)<G_1(j)}}
        (p_i-p_j)\cdot\mathds{1}\{G_2(i)>G_2(j)\} \\
    ={}& K_p(G_1,G_2).
\end{align*}

Suppose then that $\Delta_p(G_1,G_2)=K_p(G_1,G_2)$ for all guessing functions
which differ by at most $M$ adjacent transpositions. Let $G_1$ and $G_2$ be two
guessing functions which differ by $M+1$ adjacent transpositions, i.e.,
$G_2=\tau_{M+1}\circ\tau_M\circ\dots\circ\tau_1\circ G_1$. Let
$\sigma=\tau_{M}\circ\dots\circ\tau_1$. By definition,
\begin{align*}
    \Delta_p(G_1,G_2) &= \Expl[p]{G_2-G_1} \\
    &= \Expl[p]{G_2-\sigma\circ G_1} + \Expl[p]{\sigma\circ G_1 - G_1} \\
    &= \Delta_p(\sigma\circ G_1,G_2) + \Delta_p(G_1,\sigma\circ G_1).
\end{align*}
The permutations $G_2$ and $\sigma\circ G_1$ differ by a single adjacent
transposition and $\sigma\circ G_1$ and $G_1$ differ by $M$ adjacent
transpositions. Thus, by strong induction,
\begin{equation*}
    \Delta_p(G_1,G_2) = K_p(\sigma\circ G_1, G_2) + K_p(G_1, \sigma\circ G_1).
\end{equation*}
Finally, \thref{lm:additive} yields $\Delta_p(G_1,G_2) = K_p(G_1,G_2)$.
\end{proof}

Note that this result requires that $K_p(G_1,G_2)$ be signed and asymmetric.
We also have a stronger version of \thref{lm:additive}.

\begin{corollary}
$K_p(\sigma_1,\sigma_3)=K_p(\sigma_1,\sigma_2)+K_p(\sigma_2,\sigma_3)$ for
any permutations $\sigma_1$, $\sigma_2$, and $\sigma_3$ and for any
distribution $p$.
\end{corollary}
\begin{proof}
The desired result follows from \thref{th:pweight} and linearity of
expectation.
\end{proof}

When considering distributions rather than particular guessing functions, the
following corollary is more applicable.

\begin{corollary}\label{cr:delta}
For all distributions $p$ and $q$, the expected cost of guesswork under
mismatch between $p$ and $q$ is given by
\begin{equation*}
    \delta(p,q) = \min_{G_q\in\mathcal{G}_q} K_p(G_p,G_q) =
    \sum_{\substack{(i,j):\\p_i>p_j}}(p_i-p_j) \cdot\mathds{1}
    \{q_i<q_j\},
\end{equation*}
with $G_p$ being any optimal guessing function for $p$.
\end{corollary}
\begin{proof}
The first equality follows directly from \thref{th:pweight}. The second
equality follows from that fact that if $G_p$ is an optimal guessing function
for $p$, then $G_p(i)< G_p(j)$ implies that $p_i\geq p_j$. This inequality can
be made strict, as $p_i-p_j=0$ if $p_i=p_j$. Similarly, $G_q(i)> G_q(j)$
implies that $q_i\leq q_j$. This inequality can also be made strict, since if
$G_q^\ast$ is the guessing function which achieves the minimum expected cost
and if $q_i=q_j$, then $G_p(i)< G_p(j)$ implies that $G_q^\ast(i)<
G_q^\ast(j)$.
\end{proof}

\section{Guesswork and a Statistical Distance}\label{sec:statistical}

The goal of this section is to bound $\delta(p,q)$ given a bound on
$\dtv(p,q)$. This connection between guesswork under mismatch and statistical
distance is achieved through a combinatorial argument concerning the summands
of the form $(p_i-p_j)$ in the probability-weighted Kendall tau signed
divergence.  Representing each of these terms by the corresponding pair
$(i,j)$, we show how these pairs can be grouped into sets and how to bound the
sum over each set. To simplify notation, we write $k\in(i,j)$ to denote that
either $k=i$ or $k=j$.

\begin{definition}
A set $M\subset[n]\times[n]$ is a \emph{set of disjoint pairs} if $i\neq j$
for all $(i,j)\in M$ and if, for all $k\in[n]$, there is at most one
element $(i,j)\in M$ such that $k\in (i,j)$.
\end{definition}

\begin{example}\label{ex:disjoint}
These are five sets of disjoint pairs for $n=5$.
\begin{equation*}
\begin{array}{c|c|c|c|c}
 M_1  &  M_2  &  M_3  &  M_4  &  M_5  \\ \hline\hline
(2,5) & (1,3) & (2,4) & (3,5) & (1,4) \\
(3,4) & (4,5) & (1,5) & (1,2) & (2,3)
\end{array}
\end{equation*}
\end{example}

The following lemma formalizes how sets of disjoint pairs can be used to
bound part of the expected cost of guesswork under mismatch given a bound on
the total variation.
\begin{lemma}\label{lm:evens}
Let $M$ be a set of disjoint pairs and let $p$ and $q$ be distributions such
that $p_i\geq p_j$ and $q_i\leq q_j$ for all $(i,j)\in M$. If $\dtv(p,q)\leq
\epsilon$, then $\sum_{(i,j)\in M} (p_i-p_j)\leq 2\epsilon$.
\end{lemma}
\begin{proof}
Suppose that $\sum_{(i,j)\in M} (p_i-p_j) > 2\epsilon$. Then,
\begin{align}
    \dtv(p,q) &= \frac{1}{2}\sum_{k=1}^n \abs{p_k-q_k} \nonumber \\
    &\geq \frac{1}{2} \sum_{(i,j)\in M} \left(\abs{p_i-q_i} +
        \abs{q_j-p_j}\right) \label{eq:e-pos} \\
    &\geq \frac{1}{2} \sum_{(i,j)\in M}\abs{p_i-p_j + q_j-q_i} \label{eq:e-tri}
    \\
    &\geq \frac{1}{2}\sum_{(i,j)\in M}(p_i-p_j) \label{eq:e-sump} \\
    &> \epsilon, \nonumber
\end{align}
where \eqref{eq:e-pos} follows by dropping non-negative terms, \eqref{eq:e-tri}
follows from the triangle inequality, and \eqref{eq:e-sump} follows from the
assumption that $p_i\geq p_j$ and $q_i\leq q_j$ for all $(i,j)\in M$.
\end{proof}
Informally, \thref{lm:evens} says that if $G_q$ permutes a set of disjoint
pairs relative to $G_p$, then the sum over the terms $(p_i-p_j)$ corresponding
to each permuted pair $(i,j)$ can be bounded in terms of the total variation
between $p$ and $q$.

\thref{lm:evens} and some established results in combinatorics directly lead to
\thref{th:dbound}. For simplicity, the proof given here only considers the case
when the number of symbols $n$ is even.
\iflong
The proof for odd $n$ involves some additional technical details which are
deferred to \cref{sec:odds}.
\fi

\begin{theorem}\label{th:dbound}
If $p$ and $q$ are distributions on $[n]$ such that $\dtv(p,q)\leq \epsilon$,
then $\delta(p,q)\leq 2(n-1)\epsilon$.
\end{theorem}

\begin{proof}
By \thref{cr:delta},
\begin{equation}\label{eq:terms}
    \delta(p,q) = \sum_{\substack{(i,j):\\ p_i>p_j}} (p_i-p_j)\cdot
        \mathds{1}\{q_i < q_j\}.
\end{equation}
Without loss of generality, we may assume that the elements of $[n]$ are
labeled such that $p_i\geq p_j$ for all $i<j$. Each term in \eqref{eq:terms}
can be associated with the pair $(i,j)$, yielding $\binom{n}{2}$ pairs.

If $n$ is even, then we can construct a tournament design on $n$, i.e., all
$\binom{n}{2}$ possible pairs of distinct elements of $[n]$ can be arranged
into a $\frac{n}{2}\times(n-1)$ array such that every element is contained in
precisely one cell of each column \cite[\S 51.1]{CD10}. This is equivalent to
constructing $n-1$ sets of disjoint pairs $M_1,\dots,M_{n-1}$ whose union
covers all $\binom{n}{2}$ possible pairs. Given such a collection of sets, we
can write
\begin{align}
    \delta(p,q) &\leq \sum_{i<j} (p_i-p_j) \label{eq:e-sum}
    \\
    &= \sum_{r=1}^{n-1} \; \sum_{(i,j)\in M_r} (p_i-p_j) \label{eq:to-bound}.
\end{align}
Equality is achieved in \eqref{eq:e-sum} when $q_i<q_j$ for all $i<j$.
Combining this with the assumption that $\dtv(p,q)\leq\epsilon$, we can bound
\eqref{eq:to-bound} using \thref{lm:evens}, which yields
\begin{equation*}
    \delta(p,q) \leq \sum_{r=1}^{n-1} 2\epsilon = 2(n-1)\epsilon.
\end{equation*}
\end{proof}

The bound in \thref{th:dbound} is not tight for all values of $\epsilon$. In
particular, since $\delta(p,q)\leq n-1$ trivially, the bound is not meaningful
for $\epsilon>1/2$. However, the following example shows that there exist
distributions $p$ and $q$ for which $\delta(p,q)$ is arbitrarily close to
$2(n-1)\epsilon$ for $\epsilon\leq 1/n$.
\begin{example}\label{ex:optimal}
Let $p$ be the distribution given by
\begin{equation*}
    p_i = \begin{cases}
        1/n + \gamma\epsilon & i=1, \\
        1/n & 1<i<n, \\
        1/n - \gamma\epsilon & i=n,
    \end{cases}
\end{equation*}
with $0<\gamma<1$ and $\epsilon\leq 1/n$. Let $q$ be a distribution such that
$q_i<q_j$ for all $i<j$ and $\dtv(p,q)\leq \epsilon$. Such a distribution $q$
always exists since $\gamma<1$; it will be very close to uniform distribution,
but without any two elements having exactly the same probabilities. By
\thref{cr:delta},
\begin{align*}
    \delta(p,q) &= \sum_{\substack{(i,j):\\ p_i>p_j}} (p_i-p_j)\cdot
        \mathds{1}\{q_i < q_j\} \\
    &= \sum_{i=2}^n (p_1-p_i) + \sum_{j=2}^n (p_j-p_n) \\
    &= (n-1)(p_1-p_n) \\
    &= \gamma\cdot2(n-1)\epsilon.
\end{align*}
Thus, this choice of $p,q$ comes within a factor of $\gamma$ (which can be
arbitrarily close to 1) of the bound in \thref{th:dbound}.
\end{example}


\begin{figure*}[!ht]
\centering
\subfloat[]{\includegraphics[width=.50\linewidth]{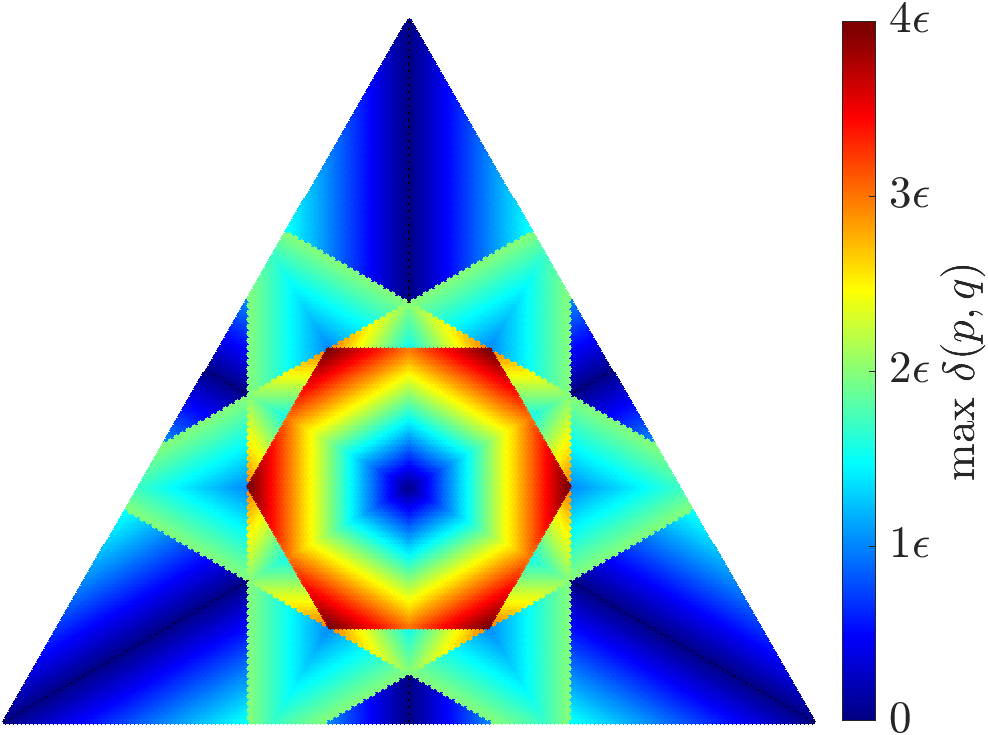}\label{fig:s-a}}
\hfil
\subfloat[]{\includegraphics[width=.50\linewidth]{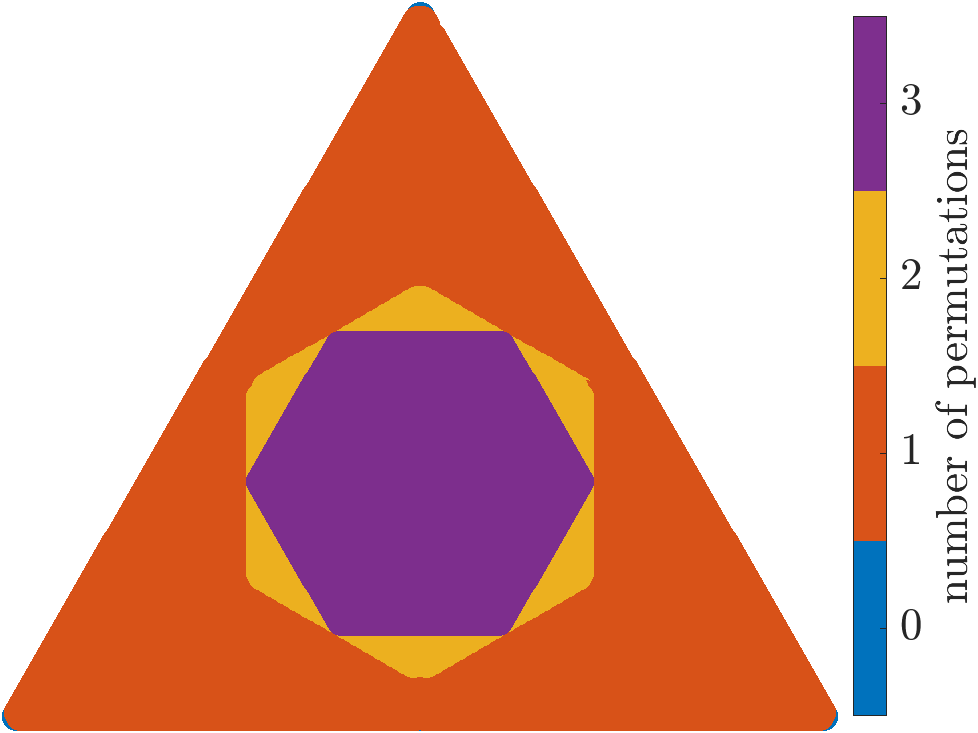}\label{fig:s-b}}
\caption{A visualization of the expected cost of guesswork under mismatch on
the 3-dimensional probability simplex. Each point on the simplex corresponds to
a different distribution over 3 symbols. The corners correspond to
distributions which assign all probability to a single symbol.  The centroid
corresponds to the uniform distribution. (a) The maximum expected cost of
guesswork under mismatch given that the mismatched distribution lies within a
total variation radius of $\epsilon=0.2$. (b) The maximum number of adjacent
transpositions between the optimal guessing function for the distribution at
each point and the optimal guessing functions for points within a total
variation radius of $\epsilon=0.2$.}
\label{fig:simplex}
\end{figure*}

This example suggests that $2(n-1)\epsilon$ is the supremum, rather than the
maximum, of $\delta(p,q)$ over distributions on $n$ symbols such that
$\dtv(p,q)\leq \epsilon$ in the regime $\epsilon\leq 1/n$. The technique used
for \thref{th:dbound} does not give a strict inequality, however, since it does
not take into account that there are multiple optimal guessing functions for
distributions which assign the same probability to multiple points. In
particular, if we consider the distributions given in \thref{ex:optimal} and
let $\gamma=1$, then the uniform distribution satisfies the total variation
bound. However, since any guessing function is optimal for the uniform
distribution, it follows that $\delta(p,q)=0$ if $q$ is uniform. Furthermore,
if $\gamma=1$, then there does not exist a $q$ within $\epsilon$ total
variation of $p$ such that $q_i<q_j$ for all $i<j$, and hence, there does not
exist any $q$ within $\epsilon$ total variation of $p$ such that
$\delta(p,q)=2(n-1)\epsilon$.

To visualize this phenomenon and conceptualize how the cost of guesswork under
mismatch within a fixed total variation radius varies, consider
\Cref{fig:simplex}, which illustrates the behavior for $n=3$. In
\Cref{fig:s-a}, we plot the maximum achievable expected cost under mismatch
within a radius of $\epsilon=0.2$ around each point on the simplex. Every such
point corresponds to a unique distribution on $[n]$, and that point is taken to
be the center around which we consider distributions for mismatch within the
radius of $\epsilon$. In \Cref{fig:s-b}, we plot the maximum number of adjacent
transpositions achievable between optimal guessing functions for the
distributions within the same radius. In three dimensions, the projection of
the $L^1$ ball onto the simplex is a regular hexagon, which gives rise to the
geometric patterns visible in \cref{fig:simplex}.

The distribution $p$ given in \thref{ex:optimal} corresponds to the most red
points in \cref{fig:s-a}. It is easy to see that, for $n=3$, there are sharp
jumps in the maximum achievable expected cost under mismatch as the maximum
achievable number of adjacent transpositions increases. A cost of $2\epsilon$
is the maximum achievable with 1 transposition, while $3\epsilon$ is the
maximum with 2, and $4\epsilon$ with 3. Although this phenomenon does appear to
some degree in higher dimensions, it is not necessarily true that the maximum
achievable expected cost under mismatch is strictly increasing in the maximum
number of achievable adjacent transpositions. In \thref{ex:optimal}, for
instance, only $2n-3$ adjacent transpositions are needed to achieve an expected
cost which is arbitrarily close to the maximum, while $n(n-1)/2$ adjacent
transpositions are possible.

\section{Conclusions and Future Work}\label{sec:conclusion}

In this paper, we analyzed the cost of guesswork under mismatch from a
non-asymptotic perspective. We showed that the total variation distance exactly
captures the maximum possible expected cost under mismatch between two
distributions if the number of symbols is considered constant. The connections
developed between permutation metrics and guesswork highlight the
combinatorial aspects of the problem and readily suggest possible extensions.

We notably only investigated the first moment of guesswork, which facilitated
the connection with the Kendall tau distance. A natural next step would be to
consider higher moments and formulate a more general framework for
combinatorial analysis. Similarly, it would be interesting to consider how
different statistical metrics or simplex geometries, such as those discussed in
\cite{NS19}, can be related to guesswork under mismatch.

\printbibliography

\clearpage\pagebreak

\appendices\crefalias{section}{appendix}

\section{An Example Concerning Definition 5}\label{sec:min}

\begin{example}
Let $n=4$ and $p$, $q$ be distributions given by:
\begin{equation*}
\begin{array}{c|cccc}
  & 1    & 2    & 3    & 4    \\ \hline
p & 0.40 & 0.30 & 0.20 & 0.10 \\
q & 0.60 & 0.20 & 0.10 & 0.10 \\
\end{array}
\end{equation*}
The distribution $p$ has a unique optimal guessing function:
\begin{equation*}
\begin{array}{c|cccc}
     & 1 & 2 & 3 & 4 \\ \hline
G_p  & 1 & 2 & 3 & 4
\end{array}
\end{equation*}
The distribution $q$ has two optimal guessing functions:
\begin{equation*}
\begin{array}{c|cccc}
     & 1 & 2 & 3 & 4 \\ \hline
G_q  & 1 & 2 & 3 & 4 \\
G_q' & 1 & 2 & 4 & 3
\end{array}
\end{equation*}
In this case, $\Delta_p(G_p,G_q)=0$ and $\Delta_p(G_p,G_q')=0.1$. Hence,
$\delta(p,q)=0$, due to the minimum in \thref{def:dist-cost}.
\end{example}

\section{Proof of Theorem 2 for odd \texorpdfstring{$n$}{n}}\label{sec:odds}

The setup for handling distributions over an odd number of symbols is largely
the same as for an even number, except for the fact that we instead
constructing an odd tournament design, yielding an $(n-1)/2 \times n$ array of
disjoint pairs. This would only give a bound of $2n\epsilon$ if we were to
apply \thref{lm:evens}.

The necessary additional detail is that, if $n$ is odd, for any set of
disjoint pairs there must exist at least one element of $[n]$ which does not
appear in any pair. In this setting, we can consider \emph{pairs of sets} which
leave out different symbols.
\begin{definition}\label{def:bridge}
Let $n$ be odd and let $M_1$ and $M_2$ be sets of disjoint pairs such that
$\abs{M_1}=\abs{M_2}=(n-1)/2$ and there exist distinct, unique elements
$k_1,k_2\in[n]$ such that $k_1\notin m_1$ for all $m_1\in M_1$ and $k_2\notin
m_2$ for all $m_2\in M_2$. The pair $(k_1,k_2)$ is defined to be the
\emph{bridge pair} for $M_1$ and $M_2$.
\end{definition}
\begin{example}\label{ex:bridge-pairs}
For the sets listed in \thref{ex:disjoint}, the pair $(2,5)$ is the bridge pair
for $M_2$ and $M_5$, while $(3,4)$ is the bridge pair for $M_3$ and $M_4$.
Indeed, the sets are labeled such that $(i,j)$ is the bridge pair for $M_i$ and
$M_j$.
\end{example}
Note that it does not make sense to consider \thref{def:bridge} when $n$ is
even, if the sets in question are not of maximal size, or if the sets in
question do not leave out distinct elements. When a bridge pair does exist,
however, it is by definition unique.

The following lemma is an analogue of \thref{lm:evens} which accounts for
bridge pairs.
\begin{lemma}\label{lm:odds}
Let $n$ be odd and let $M_1$ and $M_2$ be sets of disjoint pairs for which
there exists a bridge pair $(k_1,k_2)$. Let $M=M_1\cup M_2\cup \{(k_1,k_2)\}$.
Let $p$ and $q$ be distributions such that $p_i\geq p_j$ and
$q_i\leq q_j$ for all $(i,j)\in M$. If $\dtv(p,q)\leq\epsilon$, then
$\sum_{(i,j)\in M} (p_i-p_j) \leq 4\epsilon$.
\end{lemma}
\begin{proof}
Suppose that $\sum_{(i,j)\in M} (p_i-p_j) > 4\epsilon$. Then,
\begin{align}
    4\dtv(p,q) &= \sum_{i=1}^n \abs{p_i-q_i} + \sum_{i=1}^n \abs{p_i-q_i}
    \nonumber
    \\
    \begin{split}\label{eq:o-tri}
    &\geq \sum_{(i,j)\in M_1} \abs{p_i-p_j + q_j - q_i}  \\
    &\quad+ \sum_{(i,j)\in M_2} \abs{p_i-p_j + q_j - q_i}  \\
    &\quad+ \abs{p_{k_1}-p_{k_2}+q_{k_2}-q_{k_1}} \\
    \end{split} \\
    &\geq \sum_{(i,j)\in M} (p_i-p_j) \label{eq:o-sump} \\
    &> 4\epsilon, \nonumber
\end{align}
where \eqref{eq:o-tri} follows from the triangle inequality and the fact that,
by the definition of a bridge pair, each element of $[n]$ appears twice amongst
all the pairs in $M$, while \eqref{eq:o-sump} follows from the assumption that
$p_i\geq p_j$ and $q_i\leq q_j$ for all $(i,j)\in M$.
\end{proof}

To prove \thref{th:dbound} for odd $n$, we use a similar construction of sets
of disjoint pairs based on a tournament design. For odd $n$, we can treat one
of the sets of disjoint pairs as a set of bridge pairs and save the extra
factor of $n$ which arises in an odd tournament design by applying
\thref{lm:odds} in place of \thref{lm:evens}.

\begin{proof}[of \thref{th:dbound}, cont.]
Recall that we may assume without loss of generality that $p_i\geq p_j$ for all
$i<j$. If $n$ is odd, then we can construct an odd tournament design, i.e., all
$\binom{n}{2}$ possible pairs can be arranged into an $\frac{n-1}{2}\times n$
array such that every element is contained in at most one cell of each column
and there is exactly one element missing from each column \cite[\S 51.1]{CD10}.
This is equivalent to constructing $n$ sets of disjoint pairs $M_1,\dots,M_n$
such that there exists a distinct bridge pair between any two of them.

In particular, if we label the sets $M_1,\dots,M_n$ according to the elements
which are missing from each set, then each pair $(i,j)\in M_1$ is a bridge pair
for $M_i$ and $M_j$. It is also follows from the properties of the odd
tournament design that no two elements of $M_1$ will be bridge pairs for the
same sets. Let $M_{i,j}=M_{i}\cup M_{j}\cup\{(i,j)\}$. Then,
\begin{align}
    \delta(p,q) &\leq \sum_{i<j} (p_i-p_j) \label{eq:odd-ineq} \\
    &= \sum_{(i,j)\in M_1}\;
        \sum_{(k,l)\in M_{(i,j)}} (p_k-p_l) \label{eq:bridge-const} \\
    &\leq \sum_{(i,j)\in M_1} 4\epsilon \label{eq:lem-odds} \\
    &= 2(n-1)\epsilon \nonumber,
\end{align}
where \eqref{eq:bridge-const} follows from the bridge pair construction and
\eqref{eq:lem-odds} follows from \thref{lm:odds}, using the fact that
equality in \eqref{eq:odd-ineq} is achieved when $q_i<q_j$ for all $i<j$.
\end{proof}

\end{document}